\documentclass{article}

\usepackage[preprint,nonatbib]{neurips_2021}

\usepackage[utf8]{inputenc} 
\usepackage[T1]{fontenc} 
\usepackage{url} 
\usepackage{booktabs} 
\usepackage{amsfonts} 
\usepackage{nicefrac} 
\usepackage{microtype} 
\usepackage[dvipdfmx]{xcolor} 

\usepackage[pdftex]{graphicx}

\usepackage{latexsym}
\usepackage{amsmath}
\usepackage{cases}
\usepackage{amsthm}
\usepackage{amssymb}
\usepackage{amsfonts}
\usepackage{comment}
\usepackage{ascmac}
\usepackage{algorithm}
\usepackage[noend]{algpseudocode}
\usepackage{url}

\newcommand{\ep}{\varepsilon}

\newcommand{\R}{\mathbb{R}}
\newcommand{\VG}{{\rm VG}}

\newtheorem{theorem}{Theorem}
\newtheorem{prop}[theorem]{Proposition}

\newtheorem{rmk}[theorem]{Remark}
\newtheorem{ex}[theorem]{Example}

\title{MEGEX: Data-Free Model Extraction Attack against Gradient-Based Explainable AI}

\author{%
 Takayuki Miura \\
 NTT Social Informatics Laboratories \\
 Tokyo, Japan \\
 \texttt{takayuki.miura.br@hco.ntt.co.jp} \\
 \AND
 Satoshi Hasegawa \\ 
 NTT Social Informatics Laboratories \\
 Tokyo, Japan \\
\texttt{satoshi.hasegawa.ks@hco.ntt.co.jp} \\
 \And 
 Toshiki Shibahara \\
 NTT Social Informatics Laboratories \\
 Tokyo, Japan \\
 \texttt{toshiki.shibahara.de@hco.ntt.co.jp} \\
}

\begin{document}

\maketitle

\begin{abstract}
The advance of explainable artificial intelligence, which provides reasons for its predictions, is expected to accelerate the use of deep neural networks in the real world like Machine Learning as a Service (MLaaS) that returns predictions on queried data with the trained model. Deep neural networks deployed in MLaaS face the threat of model extraction attacks. A model extraction attack is an attack to violate intellectual property and privacy in which an adversary steals trained models in a cloud using only their predictions. In particular, a data-free model extraction attack has been proposed recently and is more critical. In this attack, an adversary uses a generative model instead of preparing input data. The feasibility of this attack, however, needs to be studied since it requires more queries than that with surrogate datasets. In this paper, we propose MEGEX, a data-free model extraction attack against a gradient-based explainable AI. In this method, an adversary uses the explanations to train the generative model and reduces the number of queries to steal the model. Our experiments show that our proposed method reconstructs high-accuracy models -- 0.97$\times$ and 0.98$\times$ the victim model accuracy on SVHN and CIFAR-10 datasets given 2M and 20M queries, respectively. This implies that there is a trade-off between the interpretability of models and the difficulty of stealing them.
\end{abstract}

\section{Introduction}
\label{sec:intro}
With the advance of deep neural networks (DNNs), artificial intelligence (AI) based on DNNs is expected to solve various social problems. As a DNN-based cloud service, Machine Learning as a Service (MLaaS) has become increasingly common. MLaaS receives queries from users and returns predictions by a deployed model. To support decision-making on the basis of the predictions in the real world, explainable AI~\cite{adadi2018peeking}, which also provides reasons for the predictions, has attracted attention. In Google Cloud, for example, gradient-based explainable AI can be used to show the importance of each element in input data.\footnote{\url{https://cloud.google.com/explainable-ai}}

The model deployed in MLaaS is crucial to its service, but the model is under threat of being stolen by adversaries. The threat is called a model extraction attack~\cite{tramer2016stealing}. In this attack, an adversary prepares input data and queries them to the target victim model. By using the predictions, the adversary trains a clone model with the same accuracy as the victim model as shown in Fig.~\ref{ms_picture}.
\begin{figure}[!t]
\centering
\includegraphics[width=0.8\linewidth]{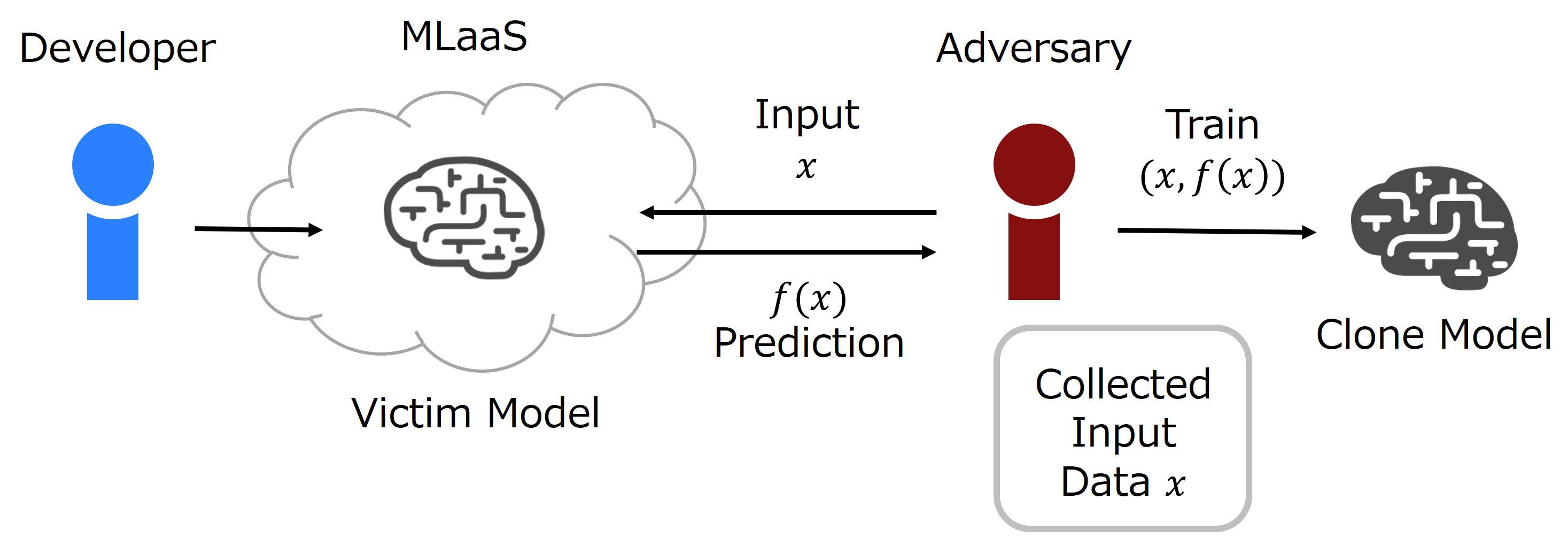}
\caption{Model Extraction Attack.}\label{ms_picture}
\end{figure}
Since companies developing DNNs spend a lot of money on training highly accurate models, they are concerned about protecting such valuable models to maintain their competitive advantages. In addition to this problem, the adversary can use the stolen model for other attacks including adversarial examples and a model inversion attack~\cite{shi2018active}.

In most model extraction attacks, adversaries must collect input data for queries. Valuable models, in fact, are often trained using the rare dataset, for instance, sales data collected by a company. In such cases, input data is difficult for the adversary even to collect. For this reason, a data-free model extraction attack has been proposed~\cite{truong2020data, kariyappa2020maze}, in which the adversary prepares a generative model to generate input data instead of collecting them (as shown in Fig.~\ref{datafreems_picture}). In a data-free model extraction attack, the adversary can steal even a valuable model that is trained with the rare dataset. On the other hand, the disadvantage of the attack is that the adversary requires more queries since the adversary has to train the generative models from scratch to generate adequate input data. Since numerous queries increase the financial cost of querying data and the possibility of being detected as abnormal users, the currently proposed data-free model extraction attacks do not seem to be severe threats. However, only a few studies have been conducted to improve the efficiency of the data-free model extraction attacks. Therefore, further studies must be conducted to clarify the threat of the attacks before DNNs are used more widely in the real world.

\begin{figure}[!t]
\begin{center}
\includegraphics[width=0.8\linewidth]{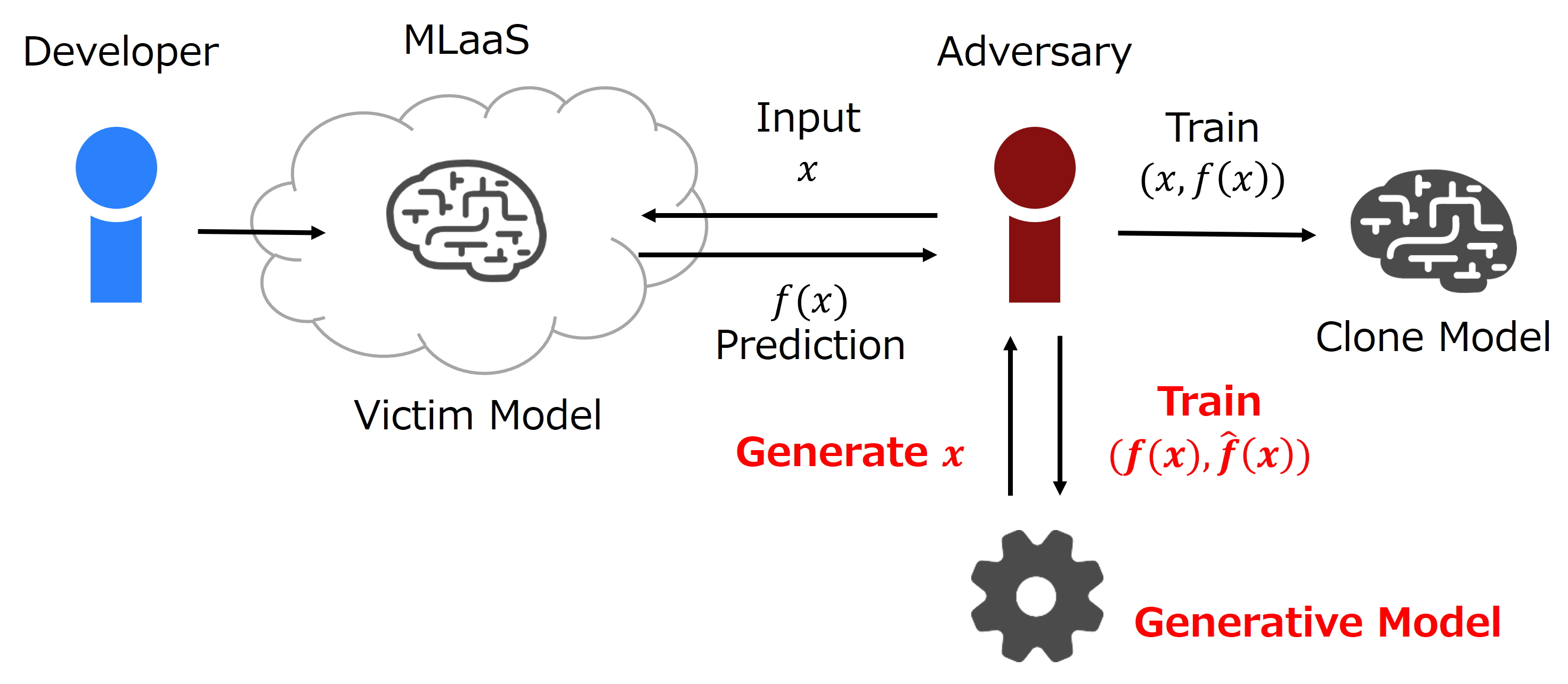}
\caption{Data-Free Model Extraction Attack.}\label{datafreems_picture}
\end{center}
\end{figure}

In this paper, we focus on explainable AI, which has become more commonly used recently to improve the usability of DNNs in the real world. We propose MEGEX, a data-free model extraction attack against gradient-based explainable AI. In MEGEX, we assume that the victim model is explainable AI with Vanilla Gradient~\cite{simonyan2013deep}. The adversary uses the explanations for computing the exact gradients to train the generative model and reduces the number of queries to 6/7 times that of the existing method. Our experiment shows that MEGEX achieves a more accurate clone model with CIFAR-10. The adversary in the existing method obtains a clone model with 80.6\% test accuracy while the adversary  in MEGEX obtains a clone model with $\ell_1$ norm loss function has 92.3\% test accuracy. This implies that there is a trade-off between the interpretability of models and the difficulty of data-free model extraction attacks.

\section{Related Work}
\label{sec:related}
\paragraph{Explainable AI} DNNs achieve high accuracy in various tasks, but their low transparency causes problems such as unfair and untrusted predictions. To improve transparency of DNNs, explainable AI is actively studied by adding post-hoc explanations~\cite{ribeiro2016should,simonyan2013deep,wachter2017counterfactual,smilkov2017smoothgrad,sundararajan2017axiomatic} and designing transparent models~\cite{NEURIPS2018_743394be,pedapati2020learning}. In image recognition tasks, gradient-based methods are commonly used to show the importance of an individual pixel~\cite{simonyan2013deep,smilkov2017smoothgrad,sundararajan2017axiomatic}. The most basic method is Vanilla Gradient~\cite{simonyan2013deep} that computes a gradient of a prediction with respect to an individual pixel. The gradient map produced by Vanilla Gradient is used as to explain which part of the image is important for the prediction. On the basis of Vanilla Gradient, many improved methods have been proposed. For example, SmoothGrad~\cite{smilkov2017smoothgrad} reduces noise in a gradient map, and Integrated Gradients~\cite{sundararajan2017axiomatic} solves the problem that not all important pixels have large gradients.

In this paper, we aim to take the first step to study data-free model extraction attacks to explainable AI by proposing an attack against VanillaGrad. Although the targeted explanation method is limited, this paper sheds light on new threats against explainable AI.

\paragraph{Model Extraction Attacks} As MLaaS becomes common, model extraction attacks that steal models in MLaaS have attracted attention~\cite{tramer2016stealing,yu2020cloudleak,yan2020cache,chandrasekaran2020exploring}. In the attacks, attackers prepare surrogate datasets in advance, query data in the datasets via MLaaS API and obtain predictions by the models in MLaaS. Finally, the attackers train clone models using queried data and predictions. Many researchers have studied the threat of model extraction attacks from practical and theoretical perspectives. Yu~et~al. proposed an attack that reduced the number of queries with transfer learning~\cite{yu2020cloudleak}. Yan~et~al. proposed a cache side channel attack to steal a DNN's architecture~\cite{yan2020cache}. Chandrasekaran~et~al. theoretically discussed the threat of model extraction attacks by formalizing them as active learning~\cite{chandrasekaran2020exploring}.

Considering the limited ability of attackers, data-free model extraction attacks have been proposed~\cite{truong2020data,kariyappa2020maze}. In this attack, attackers are assumed to have difficulty preparing surrogate datasets. For example, medical data are difficult for attackers to collect in advance. Data-free model extraction attacks are based on a knowledge distillation method~\cite{hinton2015distilling} that does not use training data. We explain this attack in detail in Section~\ref{sec:preliminary}.

\paragraph{Model Extraction Attacks to Explainable AI} There have been only a few studies on model extraction attacks against explainable AI~\cite{milli2019model,aivodji2020model}. Milli~et~al. proposed an attack that used a loss regarding the distance between a gradient-based explanation of a victim model and that of a clone model and experimentally showed that the loss improved the efficiency of the attack~\cite{milli2019model}. Ulrich~et~al. proposed an attack that used counterfactual explanations of the targeted model for training the clone model efficiently~\cite{aivodji2020model}.

In the above attacks, attackers are assumed to prepare surrogate datasets in advance. In other words, data-free model extraction attacks against explainable AI have not been proposed yet.

\section{Preliminaries}
\label{sec:preliminary}
In this section, we introduce basic notions before describing our proposed method.

\subsection{Deep Neural Networks}
In this paper, we assume DNNs as follows. Let $f : {\bf X} \to {\bf Y}$ be a DNN where ${\bf X}=\R^d$ is an input data space and ${\bf Y} := \Delta_{\ge 0}^c = \{ y \in \R^c_{\ge 0} \mid \sum_{i=1}^c y_i = 1 \}$ is a space of outputs, i.e., confidence values. A DNN $f$ classifies input data into $c$ classes. We do not assume specific DNN architectures. In other words, we assume that $f$ is a typical DNN such as ResNet for the image classification task. The only assumption is that a softmax function is used for the activation function in the last layer.

When training a DNN, we use supervised training data $(x,y) \in {\bf X} \times {\bf C}$, where $x$ is input data, $y$ is a one-hot vector of its label, and ${\bf C}\subset {\bf Y}$ is a set of one-hot vectors. An objective function is typically defined by using a loss function $l : {\bf Y}\times{\bf Y} \to \R_{\ge 0}$ as follows:
\[
\mathcal{L}(\theta_f , D) := \sum_{i=1}^N l (f(x_i), y_i) .
\]
Here $D = \{ (x_i, y_i)\}_{i=1, \ldots , N} \in ({\bf X} \times {\bf C})^N$ is a training dataset, and a vector $\theta_f$ is the parameter of the DNN $f$. The goal of the training is to find the parameter minimizing $\mathcal{L}(\theta_f , D)$. The parameter $\theta_f$ is iteratively updated on the basis of optimization methods, for instance, stochastic gradient descent (SGD). Specifically, $\mathcal{L}(\theta_f , D)$ is minimized by updating the parameter $\theta_f$ in the opposite direction of the gradient $\nabla_{\theta_f} \mathcal{L}(\theta_f , D)$.

\subsection{Gradient-Based Explainable AI}
The {\bf Vanilla Gradient}~\cite{simonyan2013deep} is a gradient-based explanation method that highlights important pixels in input images for prediction results by a DNN. Let $f : {\bf X} \to {\bf Y}$ be a DNN required to explain and $x \in {\bf X}$ be an input image. The Vanilla Gradient returns the following vector as an explanation:
\[
{\rm VG}(x) := \nabla_x f(x) \in {\bf X}^c .
\]
A large partial differential value means that change in the corresponding pixel causes a large change in the prediction results. Such a pixel can be regarded as an important pixel for the prediction results. Therefore, the map of the partial differential values is used to explain the prediction result by the DNN. 

\begin{rmk}
In this paper, we use two gradients $\nabla_{\theta_f}$ and $\nabla_x$. The former is differential values with respect to the parameter of a DNN, and the latter is those with respect to pixels of an input image. The difference between these two gradients needs to be carefully recognized.
\end{rmk}

\subsection{Data-Free Model Extraction Attacks}
\label{ms_def}
In model extraction attacks, an adversary steals a trained model by making a copy of the model using its predictions. We call the trained model $f:{\bf X} \to {\bf Y}$ a victim model and the copied model $\hat{f}:{\bf X} \to {\bf Y}$ a clone model. For the adversary, the victim model is a black-box and the adversary can only use its predictions for queries. The goal of the adversary is to obtain a high-accuracy clone model.

In most model extraction attacks, adversaries collect input data $x \in {\bf X}$ and query them to the victim model to obtain the training dataset for the clone model $\{ (x_i, f(x_i))\}$. Querying data to the victim model and training the clone model are conducted alternately.

In contrast, a data-free model extraction attack, which has been proposed recently, does not require adversaries to collect input data. In this attack, an adversary trains a generative model instead of collecting input data. The generative model is a DNN $G : \R^r \to {\bf X}$ that transforms noise drawn from a Gaussian distribution into input data. The adversary queries the generated input data $x\in{\bf X}$ to the victim model and obtains training data $(x, f(x))$. The clone model $\hat{f}$ is trained to be similar to the victim model $f$. On the other hand, the generative model is trained to generate data whose prediction by the clone model is different from that by the victim model. In other words, the generative model is trained to increase the loss function of the clone model. Since the loss function regarding such generated data is large, the parameter of the clone model is expected to be effectively updated. Although the generative model $G$ does not learn a distribution of an input data space ${\bf X}$, it learns to generate input data such that they help the clone model training at each stage.

\paragraph{Formalization} The training of the clone and the generative model is formalized as follows. Let $z \sim \mathcal{N}(0, I_r)$ be a noise drawn from Gaussian distribution and $l :{\bf Y} \times {\bf Y} \to \R$ be a loss function. Then the objective function is 
\[
\mathcal{L} := \mathcal{L}(\theta_G, \theta_{\hat{f}}, \theta_{f}, z) := -l(f(G(z)), \hat{f}(G(z))) .
\]
The corresponding optimization problem is
\begin{equation}\label{loss}
\min_{\theta_G} \max_{\theta_{\hat{f}}} \mathbb{E}_{z \sim \mathcal{N}(0, I_r)} [\mathcal{L}(\theta_G, \theta_{\hat{f}}, \theta_{f}, z)].
\end{equation}
This problem is solved by computing the gradients for the generative model $G$ and the clone model $\hat{f}$ alternately. Specifically, the adversary needs two gradients $\nabla_{\theta_{\hat{f}}} \mathcal{L}$ and $\nabla_{\theta_G} \mathcal{L}$. The adversary can compute the former by the backpropagation, but cannot compute the latter.

Hence the following approximation is usually used. First, for $x = G(z)$, set 
\begin{equation}\label{loss_part}
L(x) := l (f(x), \hat{f}(x)).
\end{equation}
Then we see $\mathcal{L} = -L \circ G(z)$. The following proposition holds, which is also essential to our proposed method described in the next section.

\begin{prop}\label{reduction}
Set $x:= G(z)$. If an adversary has the gradient $\nabla_x L(x)$, then the adversary can obtain the gradient $\nabla_{\theta_G} \mathcal{L}$. Therefore, an adversary with $\nabla_x L(x)$ is able to train the generative model.
\end{prop}
\begin{proof}
By the chain rule, we see
\begin{equation}
 \nabla_{\theta_G} \mathcal{L} = -\nabla_{\theta_G} G(z) \cdot \nabla_x L(x).
\end{equation}
The adversary can compute $\nabla_{\theta_G} G(z)$ by the backpropagation. Therefore, we see that the adversary is able to obtain $\nabla_{\theta_G} \mathcal{L}$ and train the generative model.
\end{proof}

By this proposition, we see that it is enough for the adversary to know $\nabla_x L(x)$. The adversary, however, cannot compute it in the existing data-free model extraction attacks. Thus, the adversary uses the following approximation.

\paragraph{Approximation of Gradients} In the studies by Truong~et~al.~\cite{truong2020data} and Kariyappa~et~al.~\cite{kariyappa2020maze}, the adversary uses zeroth-order gradient estimation with additional $m$ queries to obtain $\nabla_x L(x)$:
\begin{equation}\label{approx_loss}
\nabla_x L(x) \approx \frac{1}{m} \sum_{i=1}^m \frac{d(L(x + \ep u_i) - L(x))}{\ep} u_i, 
\end{equation}
where $u_1, \ldots , u_m $ are random $d$-dimensional vectors whose norms are one. In this method, for one training of the generative model, the adversary needs to query $m+1$ input images.

\paragraph{Loss Functions} The existing method~\cite{truong2020data} shows the difference in the loss function used by the adversary causes the difference in the clone model accuracy. We introduce loss functions used in the existing methods. 

In MAZE~\cite{kariyappa2020maze}, an adversary uses KL-divergence
\[
l_{{\rm KL}}(y, \hat{y}) = \sum_{i=1}^c y_{i} \log \frac{y_i}{\hat{y}_i}
\]
as a loss function. In data-free model extraction (DFME)~\cite{truong2020data}, in addition to KL-divergence, an adversary uses the L1 norm with reconstructed logits
\[
l_{\ell_1}(y, \hat{y}) = || \sigma^*(y) - \sigma^*(\hat{y}) ||_1
\]
as a loss function in order to prevent gradients from vanishing. We call this loss function the $\ell_1$ norm loss. Here a function $\sigma^* : Y \to \R^c$ is the inverse function of the softmax function under the assumption that the sum of logits is zero. This function is described as follows: \[ \sigma^*(y)_i = \log y_i - \frac{1}{c} \sum_{j=1}^c \log y_j
\]
for each $1\le i\le c$. 

\section{MEGEX: Data-Free Model Extraction Attacks to Gradient-Based Explainable AI}
\label{sec:proposal}
In this section, we propose MEGEX -- a data-free \underline{m}odel \underline{e}xtraction attack against \underline{g}radient-based \underline{ex}plainable AI.

\subsection{Overview}
Let $f : {\bf X} \to {\bf Y}$ be a victim model that outputs an explanation of Vanilla Gradient. As with DFME, the adversary trains both a clone model $\hat{f} : {\bf X} \to {\bf Y}$ and a generative model $G : \R^r \to {\bf X}$ that generates input data. In the proposed method, MEGEX, the adversary can use the gradient $\VG(x) = \nabla_x f(x)$ to train the generative model. 

\begin{figure}[!t]
\begin{center}
 \centering\includegraphics[width=0.8\linewidth]{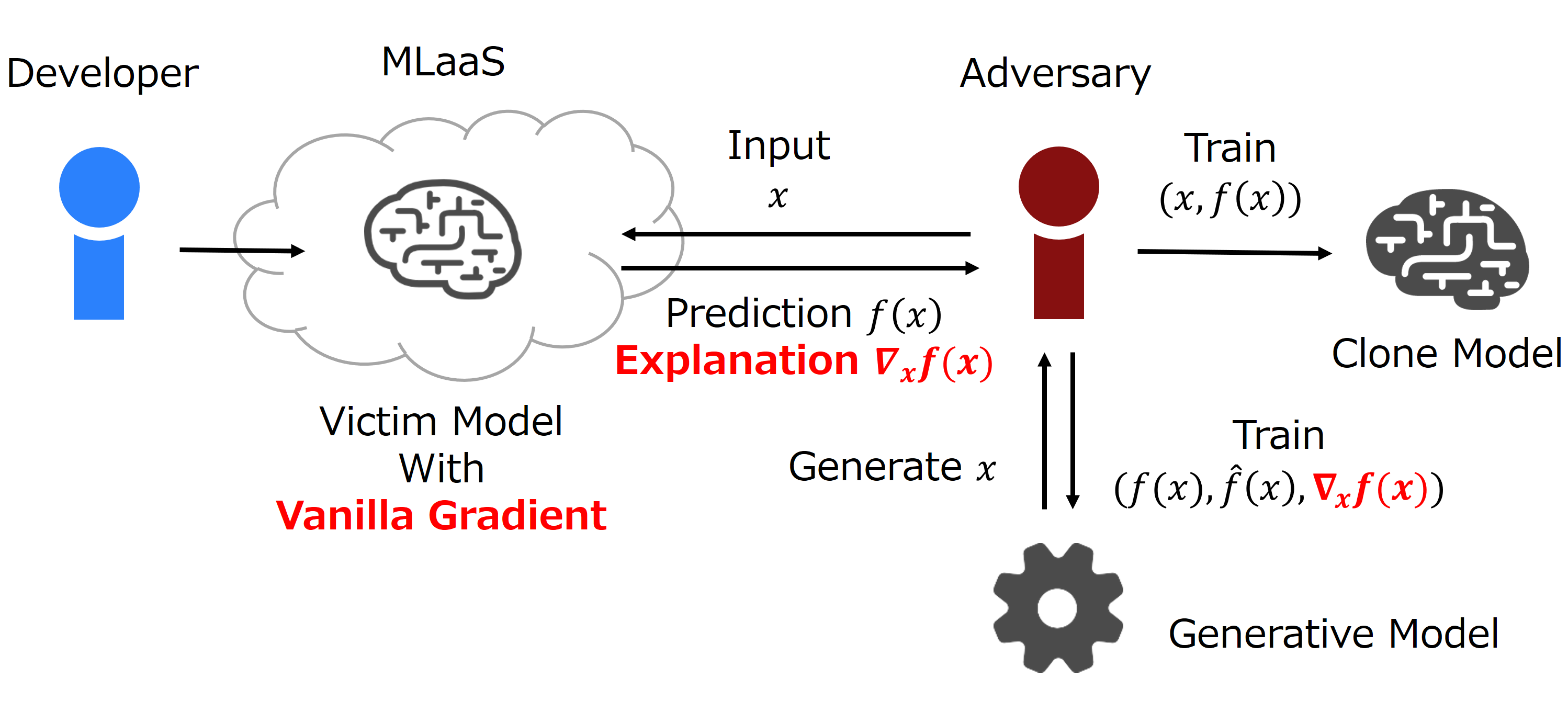}
 \caption{MEGEX: a data-free model extraction attack against gradient-based explainable AI}\label{proposal_picture}
\end{center} 
\end{figure}

The overview of the proposed method is shown in Fig.~\ref{proposal_picture} and Algorithm~\ref{proposalmethod}.

\begin{figure}[t]
 \begin{algorithm}[H]
 \caption{: MEGEX}
 \label{proposalmethod}
 \begin{algorithmic}[1]
 \Require victim model $f$, query budget $Q$, positive numbers $N_G, N_C$, learning rates $\eta_G, \eta_C$ 
 \Ensure clone model$\hat{f}$
 \State Initialize $G, \hat{f}, q \leftarrow 0$
 \While{ $q<Q$}
 \For{$i \leftarrow 1,\ldots,N_G$}
 \Comment{Training of $G$}
 \State $z \leftarrow \mathcal{N}(0, I_r)$
 \Comment{Noise drawn from Gaussian distribution}
 \State $x \leftarrow G(z)$
 \State $y \leftarrow f(x)$, $\hat{y} \leftarrow \hat{f}(x), \nabla_x f(x) \leftarrow \VG(x)$
 \State Compute $\nabla_{x} L(x) $ by $y, \hat{y}, \nabla_x f(x)$.
 \Comment{Key of the proposal}
 \State $\nabla_{\theta_G} \mathcal{L} \leftarrow \nabla_{\theta_G} G(z) \cdot \nabla_x L(x)$
 \State $\theta_G \leftarrow \theta_G - \eta_G \nabla_{\theta_G} \mathcal{L}$
 \EndFor 
 \For{$i \leftarrow 1,\ldots,N_C$}
 \Comment{Training of $\hat{f}$}
 \State $z \leftarrow \mathcal{N}(0, I_r)$
 \Comment{Noise drawn from Gaussian distribution}
 \State $x \leftarrow G(z)$
 \State $y \leftarrow f(x)$, $\hat{y} \leftarrow \hat{f}(x)$
 \State Compute $\nabla_{\theta_{\hat{f}}} \mathcal{L}$.
 \Comment{Backpropagation}
 \State $\theta_{\hat{f}} \leftarrow \theta_{\hat{f}} + \eta_C \nabla_{\theta_{\hat{f}}} \mathcal{L}$
 \EndFor
 \State $q \leftarrow q + N_G + N_C$
 \EndWhile
 	\end{algorithmic}
 \end{algorithm}
\end{figure}

\begin{enumerate}
\item Generate $N_G$ input data and query them to the victim model. Train the generative model using the predictions and explanations (lines 3-9). Here we use the explanations to compute the gradient $\nabla_{\theta_G} \mathcal{L}$.

\item Generate $N_C$ input data and query them to the victim model. Train the clone model with the pair of input data and the predictions (lines 10-15).

\item Repeat 1 and 2 until the number of queries reaches query budget $Q$.

\end{enumerate}

\subsection{Training Generative Models Using Gradient-Based Explanations}
In a data-free model extraction attack, an adversary solves the problem (\ref{loss}) by computing the gradients of the objective function with respect to parameters of a clone model and a generative model alternately. To train the clone model, it is necessary to compute the gradient $\nabla_{\theta_{\hat{f}}} \mathcal{L}$. This can be computed by the adversary using the backpropagation. On the other hand, the adversary cannot obtain the gradient $\nabla_{\theta_G} \mathcal{L}$ to train the generative model in the existing method. By Proposition~\ref{reduction}, it is enough to compute $\nabla_x L(x)$ since we see that $\nabla_{\theta_G} \mathcal{L} = -\nabla_{\theta_G} G(z) \cdot \nabla_x L(x)$. Here, by the chain rule, we have
\begin{equation}\label{loss_grad}
\nabla_x L(x) = \nabla_x f(x) \cdot \nabla_{y} l(f(x), \hat{f}(x)) + \nabla_x \hat{f}(x) \cdot \nabla_{\hat{y}} l(f(x), \hat{f}(x)),
\end{equation}
where $y$ is the first variable and $\hat{y}$ is the second variables of the loss function. In the existing method, the adversary obtains the terms other than $\nabla_xf(x)$. Now the adversary also obtains explanations with Vanilla Gradient $\VG(x) = \nabla_x f(x)$ in MEGEX. Thus, the adversary can compute $\nabla_x L(x)$ by equation (\ref{loss_grad}) exactly. In this method, the adversary can use any almost everywhere differentiable loss function to train the generative model.

\begin{ex}
If the loss function is KL-divergence, we can compute
\[
\nabla_x L(x) = \sum_{j=1}^c (1 + \log \frac{f_j (x)}{\hat{f}_j(x)}) \nabla_x f_j(x) - \frac{f_j(x)}{\hat{f}_j(x)} \nabla_x \hat{f}_j (x) .
\]

If the loss function is $\ell_1$ norm loss, by setting $v := \sigma^*(f(x))$, $s := \sigma^*(\hat{f}(x))$, $A := \frac{1}{c}\sum_{j=1}^c \frac{1}{f_j(x)}$, $\hat{A} := \frac{1}{c}\sum_{j=1}^c \frac{1}{\hat{f}_j(x)}$, we can compute
\[
\nabla_x L(x) = \sum_{i=1}^c {\rm sign}(v_i - s_i)\bigl( (\frac{1}{f_i(x)}- A)\nabla_x f_i(x) - (\frac{1}{\hat{f}_i(x)}-\hat{A})\nabla_x \hat{f}_i(x)\bigr) .
\]

\end{ex}

\subsection{Comparison of the Numbers of Queries}
In MEGEX, the adversary consumes $B \cdot (N_G + N_C)$ query budget in one iteration (from lines 3 to 15 in Algorithm~\ref{proposalmethod}). In contrast, in the existing method DFME~\cite{truong2020data} and MAZE~\cite{kariyappa2020maze}, the adversary needs to query extra $m$ input data for zeroth-order gradient estimation. Hence, the adversary consumes $B \cdot (N_G \cdot (m+1) + N_C)$ query budget in one iteration. Therefore, we see that our proposed method reduces the number of queries with the ratio
\[
\frac{N_G + N_C}{N_G \cdot (m+1) + N_C} .
\]

In accordance with DFME and MAZE, we use $N_G = 1, N_C = 5$ and a batch size $B=256$. In DFME and MAZE, $m=1$ and $m=10$, respectively. Therefore, the number of queries the proposed method needs is 6/7 times that of DFME and 3/8 times that of MAZE. In addition, the adversary in MEGEX can use the exact gradients to train the generative models. 

\section{Experimental Evaluation}
\label{sec:experiment}
We compare our proposed method (MEGEX) with the existing method (DFME).  MEGEX needs fewer queries and gives the adversary an exact gradient to train the generative model. Thus, we can expect MEGEX to achieve a more accurate clone model with fewer queries than that without explanations. In this section, we experimentally evaluate the clone model accuracy per query in both methods with two datasets: CIFAR-10 and SVHN. From the format of datasets, we see that ${\bf X} = \R^{32 \times 32 \times 3}$ and ${\bf Y} = \{y \in \R^{10}_{\ge 0} \mid \sum_{i=1}^{10} y_i = 1 \}$.

\subsection{Experimental Setup}
The implementation of MEGEX was based on that of DFME. \footnote{\url{https://github.com/cake-lab/datafree-model-extraction}} We refer to DFME~\cite{truong2020data} for neural network architectures and hyperparameters.

We evaluate the test data accuracy of clone models. For each pair, \{KL-divergence, $\ell_1$ norm loss\} $\times$ \{MEGEX, DFME\}, we train a clone model and a generative model using the victim model five times.

There were a few trials where the clone model was not trained at all. For this reason, if the accuracy of clone models stays less than 40\% at 10M queries, we stop the training and redo it.

\paragraph{Neural Network Architectures} For each dataset, we use the same architectures as those in the work by Truong~et~al.~\cite{truong2020data}. The victim model architecture is a ResNet-34 and the clone model architecture is a ResNet-18. The generative model architecture is a neural network with three convolutional layers, interleaved with linear up-sampling layers, batch normalization layers, and rectified linear unit (ReLU) activations for all layers except the last one. The last layer is the hyperbolic tangent function.

\paragraph{Training by an Adversary} For each dataset, the following are common settings. The adversary uses KL-divergence and $\ell_1$ norm loss as a loss function and also uses $N_G=1, N_C=5, B= 256$ in Algorithm~\ref{proposalmethod}. The clone model is trained with an SGD optimizer, and the generative model is trained with an Adam optimizer.

The query budgets and the learning rates are different for each dataset. The query budget $Q$ is $20{\rm M}=2.0 \times 10^7$ for CIFAR-10 and 2M for SVHN. For CIFAR-10, the adversary uses the learning rate of the generative model $\eta_G$ of $5.0 \times 10^{-4}$ and that of the clone model $\eta_C$ of 0.1 and for SVHN, that of the generative model $\eta_G$ of $5.0 \times 10^{-5}$ and that of the clone model $\eta_C$ of 0.1. In particular, for CIFAR-10, the adversary uses a learning rate scheduler that multiplies the learning rate by a factor 0.3 at $0.1 \times$ , $0.3 \times$, and $0.5 \times$ the total query budget. For SVHN, learning rate schedulers are not used.

In DFME, the adversary uses $m=1$ and $\ep = 0.001$ for zeroth-order gradient estimation.

\subsection{Experimental Results}
\begin{table}[!t]
\caption{Accuracy (normalized accuracy): Normalized accuracy is the clone accuracy divided by the victim accuracy}
 \label{result_table}
 \centering
 \begin{tabular}{lrrrrr} 
 \toprule
 Dataset (budget) & Victim & MEGEX-$\ell_1$ & MEGEX-KL & DFME-$\ell_1$ & DFME-KL \\ 
 \midrule
 CIFAR-10 (20M) & 95.5$\%$ & 92.3$\%${\small \ ($0.97\times$)} & 72.1$\%${\small \ ($0.75\times$)} & 80.6$\%${\small \ ($0.84\times$)} & 72.1$\%${\small \ ($0.75\times$)} \\
 SVHN (2M) & 94.8$\%$ & 92.5$\%${\small \ ($0.98\times$)} & 89.9$\%${\small \ ($0.95\times$)} & 92.5$\%${\small \ ($0.98\times$)} & 89.7$\%${\small \ ($0.95\times$)} \\
 \bottomrule
 \end{tabular}
\end{table}

The results of the experiments are shown in Table~\ref{result_table}. We plot the mean values of the clone model accuracy in the five experiments in Figs.~\ref{cifar10_acc} and ~\ref{svhn_acc} for each dataset. The horizontal axis is the number of queries, and the vertical axis is the test data accuracy of clone models. The error bars express the standard deviations of the results of the five experiments. The accuracy of the victim model is 95.5\% with CIFAR-10 and 94.8\% with SVHN.

\paragraph{CIFAR-10} MEGEX with $\ell_1$ norm loss achieves the test data accuracy of 72.1\% at 1M queries, 90.4\% at 10M queries, and 92.3\% at 20M queries. DFME with $\ell_1$ norm loss achieves 33.6\% at 1M queries, 74.5\% at 10M queries, and 80.6\% at 20M queries.

For CIFAR-10, our result of DFME-$\ell_1$ is worse than that of the paper by Truong~et~al.~\cite{truong2020data}, where DFME-$\ell_1$ achieves 88.1\% at 20M queries. Nevertheless, the mean of five trials of MEGEX with $\ell_1$ norm loss still outperforms the result reported in the paper. 

Both methods with KL-divergence achieve test data accuracy of $72.1\%$ at 20M queries. We consider this result comes from vanishing gradients, which is also reported by Truong~et~al.~\cite{truong2020data}. Though MEGEX uses the exact direction gradients, they are too small to train the generative model efficiently. 

\paragraph{SVHN} When the dataset is SVHN, there is no difference between the results of MEGEX and DFME. As with CIFAR-10, we also consider this result comes from vanishing gradients.

Both methods with $\ell_1$ norm loss achieve around $92\%$ at 0.5M queries. We consider that this is because the SVHN dataset is an easy task for ResNet-18 with $\ell_1$ norm. Both methods with KL-divergence have low accuracy until around $0.5$M queries and achieve around $90\%$ at $1.6$M queries.

\begin{figure}[!t]
\begin{center}
 \includegraphics[width=0.6\linewidth]{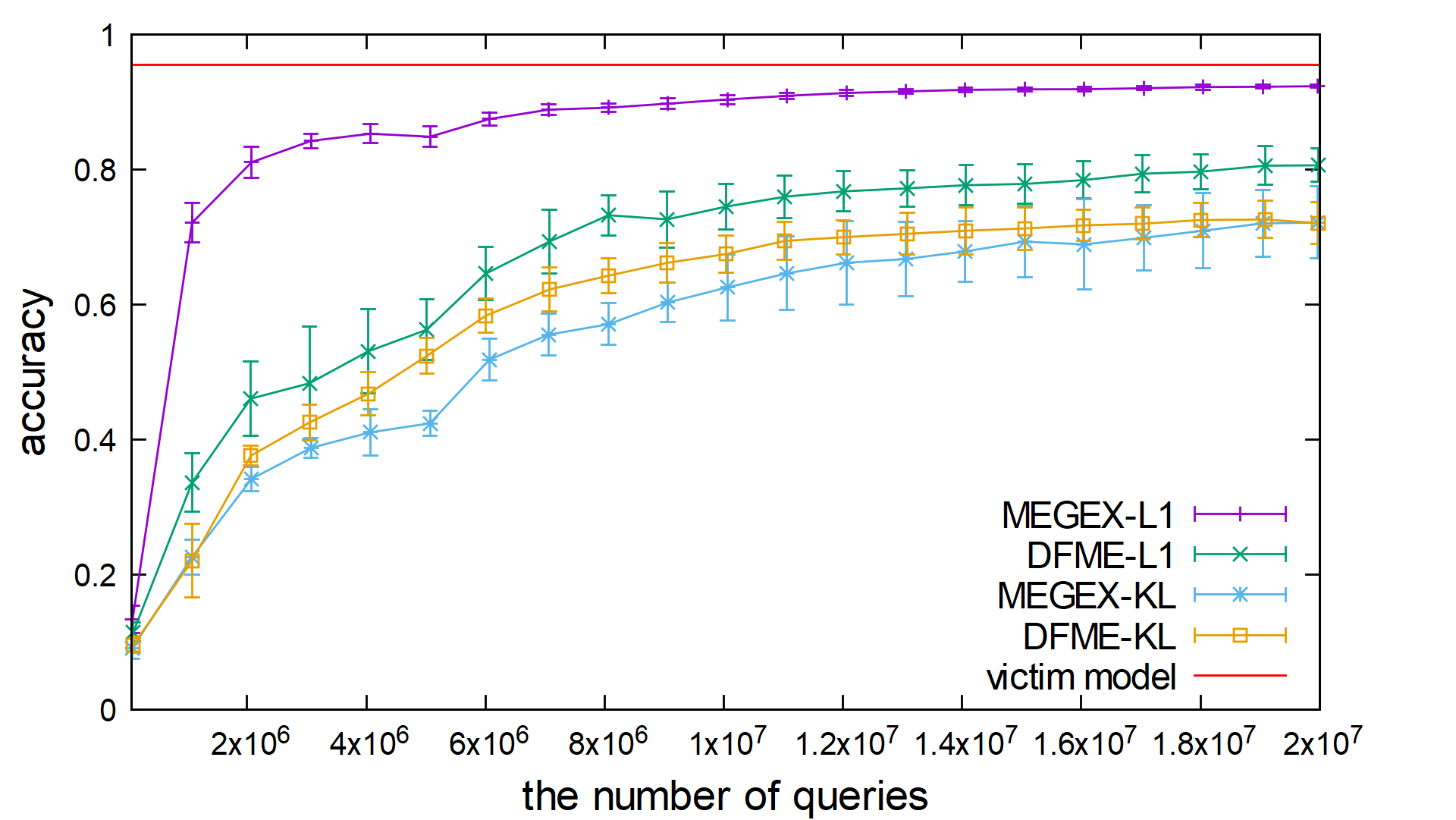}
 \caption{Test accuracy with respect to the number of queries for CIFAR-10.}\label{cifar10_acc}
\end{center}
\end{figure}

\begin{figure}[!t]
\begin{center}
 \includegraphics[width=0.6\linewidth]{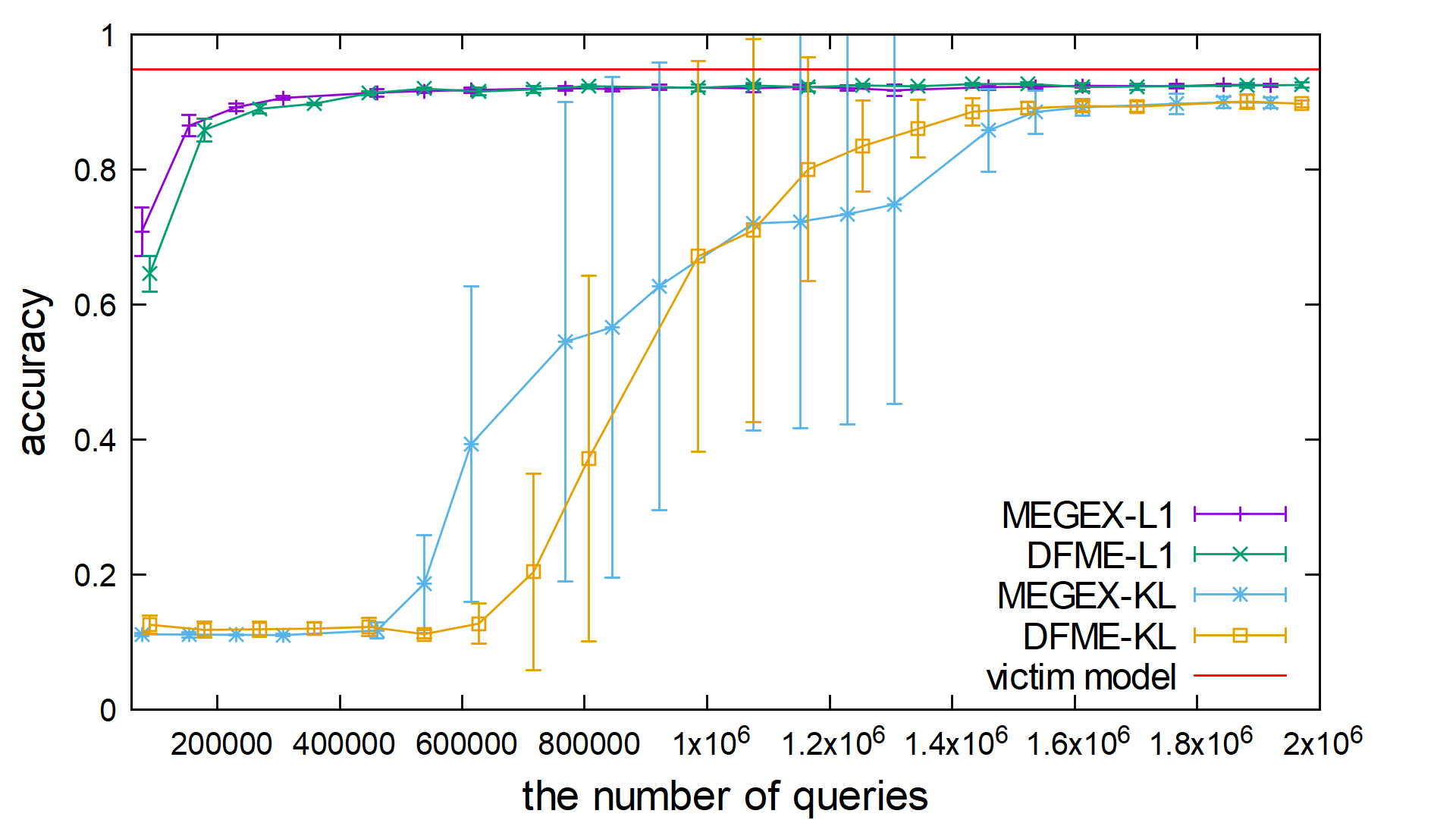} 
 \caption{Test accuracy with respect to the number of queries for SVHN.}\label{svhn_acc}
\end{center}
\end{figure}

\section{Conclusion}
\label{sec:conclusion}
In this research, we are the first to propose a data-free model extraction attack against gradient-based explainable AI and point out a new threat to explainable AI. In the proposed method, MEGEX, using the explanations, the adversary can train the generative model efficiently. Our experiments show that this method reconstructs high-accuracy models -- 0.97$\times$ and 0.98$\times$ the victim model accuracy on SVHN and CIFAR-10 datasets given 2M and 20M queries respectively. In particular, with CIFAR-10, the adversary in the existing method obtains a clone model with 80.6\% test accuracy while the adversary in MEGEX obtains a clone model  with $\ell_1$ norm loss function with 92.3\% test accuracy. This implies that there is a trade-off between the interpretability of models and the difficulty of stealing them. With the advance of machine learning, this problem may become real in the near future.

Before that, we need to investigate whether it is possible to achieve data-free model extraction attacks against other gradient-based explainable AI, like SmoothGrad and Integrated Gradient. We also have to evaluate the existing defense results~\cite{juuti2019prada, orekondy2019prediction, szyller2019dawn}, which are proposed against the existing model extraction attacks.

{\small
\bibliography{paper_arxiv_revised} 

\begin{thebibliography}{10}

\bibitem{adadi2018peeking}
Amina Adadi and Mohammed Berrada.
\newblock Peeking inside the black-box: A survey on explainable artificial
  intelligence ({XAI}).
\newblock {\em IEEE Access}, 6:52138--52160, 2018.

\bibitem{aivodji2020model}
Ulrich A{\"\i}vodji, Alexandre Bolot, and S{\'e}bastien Gambs.
\newblock Model extraction from counterfactual explanations.
\newblock {\em arXiv preprint arXiv:2009.01884}, 2020.

\bibitem{chandrasekaran2020exploring}
Varun Chandrasekaran, Kamalika Chaudhuri, Irene Giacomelli, Somesh Jha, and
  Songbai Yan.
\newblock Exploring connections between active learning and model extraction.
\newblock In {\em 29th {\rm USENIX} Security Symposium ({\rm USENIX} Security
  20)}, pages 1309--1326, 2020.

\bibitem{NEURIPS2018_743394be}
Sanjeeb Dash, Oktay Gunluk, and Dennis Wei.
\newblock Boolean decision rules via column generation.
\newblock In {\em Proceedings of the 32nd Advances in Neural Information
  Processing Systems}, pages 4655--4665.

\bibitem{hinton2015distilling}
Geoffrey Hinton, Oriol Vinyals, and Jeff Dean.
\newblock Distilling the knowledge in a neural network.
\newblock {\em arXiv preprint arXiv:1503.02531}, 2015.

\bibitem{juuti2019prada}
Mika Juuti, Sebastian Szyller, Samuel Marchal, and N~Asokan.
\newblock Prada: protecting against dnn model stealing attacks.
\newblock In {\em 2019 IEEE European Symposium on Security and Privacy
  (EuroS\&P)}, pages 512--527. IEEE, 2019.

\bibitem{kariyappa2020maze}
Sanjay Kariyappa, Atul Prakash, and Moinuddin Qureshi.
\newblock Maze: Data-free model stealing attack using zeroth-order gradient
  estimation.
\newblock {\em arXiv preprint arXiv:2005.03161}, 2020.

\bibitem{milli2019model}
Smitha Milli, Ludwig Schmidt, Anca~D Dragan, and Moritz Hardt.
\newblock Model reconstruction from model explanations.
\newblock In {\em Proceedings of the Conference on Fairness, Accountability,
  and Transparency}, pages 1--9, 2019.

\bibitem{orekondy2019prediction}
Tribhuvanesh Orekondy, Bernt Schiele, and Mario Fritz.
\newblock Prediction poisoning: Towards defenses against dnn model stealing
  attacks.
\newblock In {\em International Conference on Learning Representations}, 2019.

\bibitem{pedapati2020learning}
Tejaswini Pedapati, Avinash Balakrishnan, Karthikeyan Shanmugam, and Amit
  Dhurandhar.
\newblock Learning global transparent models consistent with local contrastive
  explanations.
\newblock {\em Proceedings of the 34th Advances in Neural Information
  Processing Systems}, 2020.

\bibitem{ribeiro2016should}
Marco~Tulio Ribeiro, Sameer Singh, and Carlos Guestrin.
\newblock "{W}hy should {I} trust you?" explaining the predictions of any
  classifier.
\newblock In {\em Proceedings of the 22nd ACM SIGKDD international conference
  on knowledge discovery and data mining}, pages 1135--1144, 2016.

\bibitem{shi2018active}
Yi~Shi, Yalin~E Sagduyu, Kemal Davaslioglu, and Jason~H Li.
\newblock Active deep learning attacks under strict rate limitations for online
  api calls.
\newblock In {\em 2018 IEEE International Symposium on Technologies for
  Homeland Security (HST)}, pages 1--6. IEEE, 2018.

\bibitem{simonyan2013deep}
Karen Simonyan, Andrea Vedaldi, and Andrew Zisserman.
\newblock Deep inside convolutional networks: Visualising image classification
  models and saliency maps.
\newblock {\em arXiv preprint arXiv:1312.6034}, 2013.

\bibitem{smilkov2017smoothgrad}
Daniel Smilkov, Nikhil Thorat, Been Kim, Fernanda Vi{\'e}gas, and Martin
  Wattenberg.
\newblock Smoothgrad: removing noise by adding noise.
\newblock {\em arXiv preprint arXiv:1706.03825}, 2017.

\bibitem{sundararajan2017axiomatic}
Mukund Sundararajan, Ankur Taly, and Qiqi Yan.
\newblock Axiomatic attribution for deep networks.
\newblock In {\em Proceedings of the 34th International Conference on Machine
  Learning}, pages 3319--3328, 2017.

\bibitem{szyller2019dawn}
Sebastian Szyller, Buse~Gul Atli, Samuel Marchal, and N~Asokan.
\newblock Dawn: Dynamic adversarial watermarking of neural networks.
\newblock {\em arXiv preprint arXiv:1906.00830}, 2019.

\bibitem{tramer2016stealing}
Florian Tram{\`e}r, Fan Zhang, Ari Juels, Michael~K Reiter, and Thomas
  Ristenpart.
\newblock Stealing machine learning models via prediction apis.
\newblock In {\em 25th {\rm USENIX} Security Symposium ({\rm USENIX} Security
  16)}, pages 601--618, 2016.

\bibitem{truong2020data}
Jean-Baptiste Truong, Pratyush Maini, Robert Walls, and Nicolas Papernot.
\newblock Data-free model extraction.
\newblock {\em arXiv preprint arXiv:2011.14779}, 2020.

\bibitem{wachter2017counterfactual}
Sandra Wachter, Brent Mittelstadt, and Chris Russell.
\newblock Counterfactual explanations without opening the black box: Automated
  decisions and the gdpr.
\newblock {\em Harv. JL \& Tech.}, 31:841, 2017.

\bibitem{yan2020cache}
Mengjia Yan, Christopher~W Fletcher, and Josep Torrellas.
\newblock Cache telepathy: Leveraging shared resource attacks to learn {\rm
  dnn} architectures.
\newblock In {\em 29th {\rm USENIX} Security Symposium ({\rm USENIX} Security
  20)}, pages 2003--2020, 2020.

\bibitem{yu2020cloudleak}
Honggang Yu, Kaichen Yang, Teng Zhang, Yun-Yun Tsai, Tsung-Yi Ho, and Yier Jin.
\newblock Cloudleak: Large-scale deep learning models stealing through
  adversarial examples.
\newblock In {\em Proceedings of Network and Distributed Systems Security
  Symposium (NDSS)}, 2020.

\end{thebibliography}
\bibliographystyle{plain}
}

\end{document}